\documentclass[runningheads,a4paper]{llncs}

\usepackage{amssymb}
\usepackage{amsmath}
\usepackage{graphicx}
\usepackage{mathpartir}
\usepackage[english]{babel}
\usepackage[T1]{fontenc}
\usepackage[utf8]{inputenc}

\usepackage{url}

\usepackage{hyperref}
\hypersetup{
  bookmarks=true,
  colorlinks=true,
  linkcolor=black,
  citecolor=black,
}

\urldef{\mailsc}\path|something@somewhere.top|
\newcommand{\keywords}[1]{\par\addvspace\baselineskip
\noindent\keywordname\enspace\ignorespaces#1}

\newcommand{\UnivOf}[1]{{\mathcal{U}}_{#1}}

\newcommand{\Inrec}{{\mathsf{R}}}

\newcommand\cT{\mathcal{T}}

\newcommand{\cA}{{\mathcal{A}}}

\newcommand{\cV}{{\mathcal{V}}}

\newcommand{\cC}{{\mathcal{C}}}

\newcommand{\limplies}{\Rightarrow}
\newcommand{\lequiv}{\Leftrightarrow}
\newcommand\mizarsize\footnotesize

\usepackage[final]{listings}
\lstdefinelanguage{Mizar}%
{columns=fullflexible,
keywords={scheme,schemes,environ,provided,where,
theorem,definition,radix,reserve,properties,struct,inhabited,expandable,attribute,
adjective,registration,coherence,defpred,cluster,from,sch,%
given,such,that,%
reflexivity, irreflexivity, symmetry, asymmetry, connectedness,and,attr,%
    antonym,existence,uniqueness,commutativity,idempotence,synonym,notation,%
    mode, means,func,pred, pred,equals,it,of,non,is,axiomatization,sethood,reconsider,redefine,%
    if,otherwise,proof,for,ex,being,holds,def,let,consider,take,not,contradiction,st,the,%
    be,thus,implies,assume,then,not,by,or,hence,thesis,end,iff,;,:,",",\#,as,qua},%
   sensitive=true,
   basicstyle={{\linespread{1.}\usefont{T1}{lmss}{m}{n}
}
},%
 keywordstyle={
 \usefont{T1}{lmss}{sbc}{n}
\selectfont
},%
  mathescape = true,
   morecomment=[l][\texttt]{::},
   literate={'}{{$\strut\mkern6mu\strut$}}1%
   {&}{{\usefont{T1}{lmss}{sbc}{n}\selectfont\texttt{\&}}}1%
   {-}{{\usefont{T1}{lmss}{m}{n}\selectfont\texttt{-}}}1%
   {->}{{$\rightarrow$}}1%
   {\{\}}{{\{\!\}}}1
  }

\makeatletter
\def\miz{\lstinline}
\makeatother

\begin{document}

\vspace{4cm}
\begin{center}
The final publication is available at Springer via
\url{http://dx.doi.org/10.1007/978-3-030-23250-4_4}
\end{center}

\mainmatter  

\title{A Tale of Two Set Theories} 

\titlerunning{A Tale of Two Set Theories}

%
%
\author{Chad E. Brown\inst{1}
  \and
  Karol P\k{a}k\inst{2}\orcidID{{0000-0002-7099-1669}}
}
\authorrunning{Chad E. Brown \and Karol P\k{a}k}

\institute{Czech Technical University in Prague \and University of Bia\l{}ystok \url{pakkarol@uwb.edu.pl}}

%
%

\maketitle

\begin{abstract}
  We describe the relationship between two versions
  of Tarski-Grothendieck set theory: the first-order set theory
  of Mizar and the higher-order set theory of Egal.
  We show how certain higher-order terms and propositions in
  Egal have equivalent first-order presentations.
  We then prove Tarski's Axiom A (an axiom in Mizar) in Egal
  and construct a Grothendieck Universe operator (a primitive with axioms in Egal)
  in Mizar.
\end{abstract}

\lstset{language=Mizar}

\section{Introduction}\label{sec:intro}

We compare two implemented versions of Tarski-Grothendieck (TG) set theory.
The first is the first-order TG implemented in Mizar~\cite{DBLP:conf/mkm/BancerekBGKMNPU15,JFR1980}
axiomatized using Tarski's Axiom A~\cite{Tarski1938,Trybulec89}.
The other is the higher-order TG implemented in Egal~\cite{EgalManual}
axiomatized using Grothendieck universes~\cite{SGA4}.
We discuss what would be involved porting Mizar developments into Egal
and vice versa.

We use Egal's Grothendieck universes (along with a choice operator)
to prove Tarski's Axiom A in Egal.
Consequently the Egal counterpart
of each of Mizar's axioms is provable in Egal
and so porting from Mizar to Egal should always be possible in principle.
In practice one would need to make Mizar's implicit reasoning using its
type system explicit, a nontrivial task outside the scope of this paper.

Porting from Egal to Mizar poses two challenges.
One is that many definitions and propositions in Egal make use of higher-order quantifiers.
In order to give a Mizar counterpart, it is enough to give a first-order reformulation
and prove the two formulations equivalent in Egal.
While this will not always be possible in principle, it has been possible
for the examples necessary for this paper.
The second challenge is to construct a Grothendieck universe operator in Mizar
that satisfies the properties of a corresponding operator in Egal.
We have constructed such an operator.

We give a brief introduction to Mizar and its version of first-order Tarski-Grothendieck in Section~\ref{sec:mizar}.
In Section~\ref{sec:egal} we introduce the new system Egal
and describe its version of higher-order Tarski-Grothendieck.
In Section~\ref{sec:hofo} we give a few examples of definitions and propositions
in Egal that can be reformulated in equivalent first-order forms.
These first-order versions have counterparts in Mizar.
Section~\ref{sec:axioma} discusses the Egal proof of Tarski's Axiom A.
In Section~\ref{sec:grothuniv}
we discuss the construction of a Grothendieck universe operator in Mizar.\footnote{At {\url{http://grid01.ciirc.cvut.cz/~chad/twosettheories.tgz}} one can find Egal, the Egal formalization files and the Mizar formalization files.}
Possibilities for future work are discussed in Section~\ref{sec:futurework}.

%

\section{Mizar and FOTG}\label{sec:mizar}

The Mizar system~\cite{Grabowski2015FDM} from its beginning aimed to
create a proof style that simultaneously
imitates informal mathematical proofs as much as possible and
and can be automatically verified to be logically correct.
A quite simple and intuitive reasoning formalism
and an intuitive soft type system play a major role
in the pursuit of Mizar’s goals.

The Mizar proof style is mainly inspired
by Ja\'s{}kowski~\cite{Jaskowski34} style of natural deduction
and most statements correspond to valid first-order
predicate calculus formulas. 
Over time the Mizar community has also added
support for syntax that goes beyond traditional first-order terms and formulas.
In particular, Mizar supports \miz{schemes} with predicate and function variables,
sufficient to formulate the Fraenkel replacement as one axiom in Mizar.
This axiom is sufficient to construct the set comprehension
$\{Fx|x \in X, Px\}$ (called {\it Fraenkel terms}) for a given set $X$, function $F$ and predicate $P$
in the Mizar language but it is impossible to define such a functor
for arbitrary $X$, $F$, $P$.
Therefore, in response to the needs of Mizar's users,
support for Fraenkel terms has been built into the system.
In fact Mizar supports a generalized notation
where the set membership relation $x \in X$ in the Fraenkel term has been replaced
by the type membership $x:\Theta$ if the Mizar type $\Theta$ has the \miz{sethood} property.
A Mizar type has the \miz{sethood} property if the collection of all objects of the type
forms a set (as opposed to a class).
Semantically, Mizar types are simply
unary first-order predicates over sets that can be parameterized by sets.
However, the type inference mechanisms make Mizar significantly more powerful and user-friendly.
The rules available for automatic type inference are influenced
by the author of a given script by
choosing the \miz{environ} (i.e., environment, see~\cite{JFR1980}).
By skillfully choosing the environment, an author can make a Mizar
article more concise and readable since the type system will
handle many inferences implicitly.
Mizar types must be inhabited and
this obligation must be
proven by a user directly in the definition of a given type
or before the first use if a type has the form of intersection of types.

Parallel to the system development, the Mizar community puts a significant effort into building
the Mizar Mathematical Library (MML)\cite{MML2017}.
The MML is the comprehensive repository of currently formalized mathematics in the Mizar system.
The foundation of the library, up to some details discussed below,
is first-order Tarski-Grothendieck set theory (FOTG).
This is a non-conservative extension of Zermelo–Fraenkel set theory (ZFC),
where the axiom of infinity has been replaced by Tarski's Axiom A.
Axiom A states that for every set $N$ there is a Tarski universe $M$
such that $N\in M$. A Tarski universe is essentially a set closed under
subsets and power sets with the property that every subset of the universe
is either a member of the universe or equipotent with the universe.
The statement of Axiom A in Mizar is shown in Figure~\ref{fig:mizartarskia}.

\begin{figure}
\vspace{-0.2cm}
\begin{lstlisting}
  reserve'N,M,X,Y,Z'for'set;
  theorem'::'TARSKI_A:1
    ex'M'st'N'in'M'&
       (for'X,Y'holds'X'in'M'&'Y'c='X'implies'Y'in'M)'&
       (for'X'st'X'in'M'ex'Z'st'Z'in'M'&'for'Y'st'Y'c='X'holds'Y'in'Z)'&
       (for'X'holds'X'c='M'implies'X,M'are_equipotent'or'X'in'M);
\end{lstlisting}
\vspace{-0.1cm}
\caption{Tarski's Axiom A in Mizar}\label{fig:mizartarskia}
\end{figure}

FOTG was not the only foundation considered for the library.
One of the main reasons it was chosen is the usefulness of Axiom A
in the formalization of category theory.
Namely, FOTG provides many universes
that have properties analogous to those of a class of all sets.
In particular, every axiom of ZFC remains true if we relativize quantifiers
to the given universe.

The axiom of choice can be proven in FOTG.
In fact Axiom A was used to prove Zermelo's well-ordering theorem
and the axiom of choice in an early MML article~\cite{wellord2.abs}.
Later changes to Mizar also yielded the axiom of choice in a more direct way
and we briefly describe the relevant changes.

While working with category theory in the Mizar system,
new constructions called
{\it permissive} definitions were introduced (implemented in Mizar-2 in the 80's~\cite{Grabowski2015FDM}).
Permissive definitions allow an author to make definitions
under assumptions where these assumptions
can be used to justify the obligations.
For example, the type \miz{morphism of a,b}
can be defined
under the assumption that there exists a morphism from \miz{a} to \miz{b}.
Without the assumption the definition of \miz{morphism of a,b} would not
be allowed since the type would not be provably inhabited (see~\cite{CKKP-CICM/MKM17,JAR2018}).

In contrast to Fraenkel terms,
permissive definitions do not have an obvious semantic justification in FOTG.
For any type \miz{$\Theta$'of'a,b,...} (depending on objects $a,b,\ldots$)
a permissive definition can be used to obtain a choice operator for
the type in the following way:\vspace{-0.05cm}%
\begin{lstlisting}
  definition
    let'a,b,...'such'that'C:'contradiction;
    func'choose(a,b,...)'->'$\Theta$'of'a,b,...'means'contradiction;
    existence'by'C; uniqueness'by'C;
  end;
\end{lstlisting}\vspace{-0.05cm}%
The definition states that given objects $a,b,\ldots$ (of appropriate types),
the function $\miz{choose}$ will return an object of type
\miz{$\Theta$'of'a,b,...}
satisfying the condition \miz{contradiction}.
The definition is made under the extra assumption \miz{contradiction}
and it is this extra assumption (from which everything can be proven)
that guarantees existence and uniqueness of an object of type
\miz{$\Theta$'of'a,b,...}
satisfying the otherwise impossible condition.
After the definition is made, Mizar allows the user to make use of
the term
\miz{choose(a,b,...)}
of type
\miz{$\Theta$'of'a,b,...}
even in non-contradictory contexts.

To avoid repetition of definitions like \miz{choose}, in 2012,
the Mizar syntax was extended by the {\it explicit} operator \miz{the}
(e.g., \miz{the'$\Theta$'of'a,b,...'}).
This new operator behaves similarly to a Hilbert $\varepsilon$-operator,
which corresponds to having a global choice operator on the universe of sets
(cf. p. 72 of~\cite{Fraenkel1973}).
ZFC extended with a global choice operator is known to be conservative
over ZFC~\cite{Felgner1971}.
The situation with FOTG is analogous to that of ZFC, and
we conjecture
FOTG extended with a global choice operator (\miz{the}) is
conservative over FOTG.
Regardless of the truth of this conjecture,
we take the proper foundation of the MML to be
FOTG extended with a global choice operator (see~\cite{JAR2018}).

\section{Egal and HOTG}\label{sec:egal}

Egal~\cite{EgalManual} is a proof checker for
higher-order Tarski-Grothendieck (HOTG) set theory.
Since this is the first publication describing Egal,
we begin by placing the system in context and
discussing various design decisions.

The idea of combining higher-order logic and set theory
is not new~\cite{Gordon1996,KirstSmolka:2018:Categoricity,ObuaHOLZF}.
However, many of the features of existing higher-order systems
(e.g., the ability to define new type constructors such as
$\alpha\times\beta$)
should in principle no longer be needed if one is doing
higher-order set theory.
Instead the higher-order logic only needs to be expressive
enough to bootstrap the set theory.
Once enough set theory has been developed
users would work with products of sets
(instead of products of types).
With this in mind, Egal begins with a ``higher-order logic''
restricted to
a simple type theory mostly in the style of Church~\cite{Church40},
extended with limited prefix polymorphism (discussed below).

Another motivation to use as restricted a form of higher-order
logic as possible is to ensure Egal satisfies the
de Bruijn criterion~\cite{BarendregtWiedijk05}:
Egal proofs should be checkable by independent small proof checkers.
For this reason
Egal places an emphasis on proof terms and proof checking.
Proof terms are $\lambda$-calculus terms corresponding
to natural deduction proofs in the Curry-Howard sense.
Egal proof scripts are presented in a way
similar to Coq~\cite{Bertot08Coq}
and instruct Egal how to construct a proof term.
Since the underlying logic is relatively simple
and the additional set theory axioms are few,
the portion of the code that does type checking and proof checking
is reasonably short.
Of course the Egal code consists of more than just a checker.
For example, the code includes
a parser allowing users to give terms using mathematical
notation and variable names (instead of the de Bruijn indices used
internally)
as well as an interpreter for proof script steps.
Nevertheless we claim Egal satisfies the de Bruijn criterion
in the sense that a small independent checker could easily be written to
take as input
serialized versions of the internal representations of Egal types, terms and proof terms
and check correctness of a sequence of definitions and proofs.
The de Bruijn criterion also provides a major point of contrast between
Egal and Mizar, as constructing an independent checker for Mizar proofs
would be nontrivial for several reasons (e.g., the soft typing system).


The kernel of the Egal system includes simply typed $\lambda$-calculus
with a type of propositions
along with a $\lambda$-calculus for proof terms.
There is a base type of individuals $\iota$ (thought of as sets),
a based type of propositions $o$ and function types $\sigma\to \tau$.
Egal also allows the use of type variables
for some purposes (e.g., defining equality or giving axioms such as functional
extensionality). To simplify the presentation, we will assume
there are no type variables at first and then briefly describe
how type variables are treated.
Without extra axioms, the logic of Egal is intentional intuitionistic higher-order logic.
On top of this logic we add constants and axioms that yield an extensional classical higher-order
set theory.

To be precise let $\cT$ be the set of types generated freely via
the grammar $o|\iota|\sigma\to\tau$.
We use $\sigma, \tau$ to range over types.
For each $\sigma\in\cT$ let $\cV_\sigma$ be a countably infinite set of variables
and assume $\cV_\sigma\cap\cV_\tau = \emptyset$ whenever $\sigma\not=\tau$.
We use $x,y,z,X,Y,f,g,p,q,P,Q,\ldots$ to range over variables.
For each $\sigma\in\cT$ let $\cC_\sigma$ be a set of constants.
We use $c, c_1, c_2,\ldots$ to range over constants.
We consider only a fixed family of constants given as follows:\vspace{-0.05cm}%
\begin{itemize}
\item $\varepsilon_{\sigma}$ is a constant in $\cC_{(\sigma\to o)\to\sigma}$ for each type $\sigma$.
\item ${\mathsf{In}}$ is a constant in $\cC_{\iota\to\iota\to o}$.
\item ${\mathsf{Empty}}$ is a constant in $\cC_{\iota}$.
\item ${\mathsf{Union}}$ is a constant in $\cC_{\iota\to\iota}$.
\item ${\mathsf{Power}}$ is a constant in $\cC_{\iota\to\iota}$.
\item ${\mathsf{Repl}}$ is a constant in $\cC_{\iota\to (\iota\to\iota)\to \iota}$.
\item ${\mathsf{UnivOf}}$ is a constant in $\cC_{\iota\to\iota}$.
\end{itemize}\vspace{-0.05cm}%
No other constants are allowed. We assume none of these constants are variables.

We next define a family $(\Lambda_\sigma)_{\sigma\in \cT}$ of typed terms as follows.
We use $s$, $t$ and $u$ to range over terms.\vspace{-0.05cm}
\begin{itemize}
\item If $x\in\cV_\sigma$, then $x\in \Lambda_\sigma$.
\item If $c\in\cC_\sigma$, then $c\in \Lambda_\sigma$.
\item If $s\in\Lambda_{\sigma\to\tau}$ and $t\in\Lambda_\sigma$, then $(st)\in\Lambda_\tau$.
\item If $x\in\cV_\sigma$ and $t\in\Lambda_\tau$, then $(\lambda x.t)\in \Lambda_{\sigma\to\tau}$.
\item If $s\in\Lambda_o$ and $t\in\Lambda_o$, then $(s\limplies t)\in\Lambda_o$.
\item If $x\in\cV_\sigma$ and $t\in\Lambda_o$, then $(\forall x.t)\in \Lambda_o$.
\end{itemize}\vspace{-0.05cm}
Each member of $\Lambda_\sigma$ is a {\emph{term of type $\sigma$}}.
Terms of type $o$ are also called {\emph{propositions}}.
We sometimes use $\varphi$, $\psi$ and $\xi$ to range over propositions.
It is easy to see that $\Lambda_\sigma$ and $\Lambda_\tau$ are disjoint for $\sigma\not=\tau$.
That is, each term has at most one type.

We omit parentheses when possible, with application associating to the left and implication associating to the right: $stu$ means $((st)u)$ and $\varphi\limplies \psi\limplies \xi$ means $(\varphi\limplies (\psi\limplies \xi))$. Binders are often combined: $\lambda x y z.s$ means $\lambda x.\lambda y.\lambda z.s$ and $\forall x y z.\varphi$ means $\forall x.\forall y.\forall z.\varphi$. To present the types of variables concisely,
we sometimes annotate variables in binders with their types,
as in $\lambda x\!:\!\sigma.s$ to assert $x\in\cV_\sigma$.
When the type of a variable is omitted entirely, it is $\iota$.

Although the only logical connectives as part of the definition of terms are implication
and universal quantification, it is well-known how to define the other connectives and quantifiers
in a way that even works in an intuitionistic setting~\cite{Brown2015}.
For this reason we freely write propositions $(\neg\varphi)$,
$(\varphi\land\psi)$,
$(\varphi\lor\psi)$,
$(\varphi\lequiv\psi)$,
$(\exists x.\varphi)$
and $(s = t)$ (for $s,t\in\Lambda_\sigma$).
Again, we omit parentheses and use common binder abbreviations
in obvious ways.

We also use special notations for terms built using the constants.
We write $s\in t$ for ${\mathsf{In}}~s~t$.
We write $\forall x\in s.\varphi$ for $\forall x.x\in s\limplies \varphi$
and $\exists x\in s.\varphi$ for $\exists x.x\in s\land \varphi$.
We write $\varepsilon x:\sigma.\varphi$ for $\varepsilon_{\sigma} (\lambda x:\sigma.\varphi)$
and $\varepsilon x\in s.\varphi$ for $\varepsilon x.x\in s\land \varphi$.
We also write $\emptyset$ for ${\mathsf{Empty}}$,
$\bigcup s$ for ${\mathsf{Union}}~s$, $\wp s$ for ${\mathsf{Power}}~s$,
$\{s|x\in t\}$ for ${\mathsf{Repl}}~t~(\lambda x.s)$
and $\UnivOf{s}$ for ${\mathsf{UnivOf}}~s$.

In general new names can be introduced to abbreviate terms of a given type.
In many cases we introduce new corresponding notations as well.
The following abbreviations are used in the statements of the axioms below:
\begin{itemize}
\item ${\mathsf{TransSet}}:\iota\to o$ is $\lambda U.\forall X\in U. X\subseteq U$. Informally we
  say {\emph{$U$ is transitive}} to mean ${\mathsf{TransSet}}~U$.
\item ${\mathsf{Union\_closed}}:\iota\to o$ is $\lambda U.\forall X\in U.\bigcup X\in U$.
  Informally we say {\emph{$U$ is $\bigcup$-closed}} to mean ${\mathsf{Union\_closed}}~U$.
\item ${\mathsf{Power\_closed}}:\iota\to o$ is $\lambda U.\forall X\in U.\wp X\in U$.
  Informally we say {\emph{$U$ is $\wp$-closed}} to mean ${\mathsf{Power\_closed}}~U$.
\item ${\mathsf{Repl\_closed}}:\iota\to o$ is $\lambda U.\forall X\in U.\forall F:\iota\to\iota.(\forall x\in X.Fx\in U)\limplies \{Fx|x\in X\}\in U$.
  Informally we say {\emph{$U$ is closed under replacement}} to mean ${\mathsf{Repl\_closed}}~U$.
\item ${\mathsf{ZF\_closed}}:\iota\to o$ is $\lambda U.{\mathsf{Union\_closed}}~U\land{\mathsf{Power\_closed}}~U\land{\mathsf{Repl\_closed}}~U$.
  Informally we say {\emph{$U$ is ZF-closed}} to mean ${\mathsf{ZF\_closed}}~U$.
\end{itemize}

The deduction system for Egal includes a set $\cA$ of closed propositions we call axioms.
The specific members of the set $\cA$ are as follows:\vspace{-0.05cm}%
\begin{description}
\item[Prop. Ext.] $\forall P Q:o.(P\lequiv Q) \limplies P = Q$,
\item[Func. Ext.] $\forall f g:\sigma\to\tau.(\forall x:\sigma.f x = g x)\limplies f = g$ (for types $\sigma$ and $\tau$),
\item[Choice] $\forall p:\sigma\to o.\forall x:\sigma.px\limplies p(\varepsilon x:\sigma.p x)$ (for each type $\sigma$),
\item[Set Ext.] $\forall X Y.X\subseteq Y \limplies Y \subseteq X \limplies X = Y$,
\item[$\in$-Induction] $\forall P:\iota\to o.(\forall X.(\forall x\in X.Px)\limplies PX)\limplies \forall X.PX$,
\item[Empty] $\neg\exists x.x\in \emptyset$,
\item[Union] $\forall X x.x\in\bigcup X\lequiv \exists Y.x\in Y \land Y\in X$,
\item[Power] $\forall X Y.Y\in\wp X\lequiv Y \subseteq X$,
\item[Replacement] $\forall X.\forall F:\iota\to\iota.\forall y.y\in \{Fx|x\in X\}\lequiv \exists x\in X.y = F x$,
\item[Universe In] $\forall N.N\in \UnivOf{N}$,
\item[Universe Transitive] $\forall N.{\mathsf{TransSet}}~\UnivOf{N}$,
\item[Universe ZF closed] $\forall N.{\mathsf{ZFclosed}}~\UnivOf{N}$ and
\item[Universe Min] $\forall N U.N\in U\limplies {\mathsf{TransSet}}~U\limplies {\mathsf{ZFclosed}}~U\limplies \UnivOf{N} \subseteq U$.
\end{description}\vspace{-0.05cm}%
The axiom set would be finite if it were not for functional extensionality and choice.
In the implementation type variables are used to specify functional extensionality
and choice. Again, we delay discussion of type variables for the moment.

The notions of free and bound variables are defined as usual,
as is the notion of a variable $x$ being free in a term $s$.
We consider terms equal up to bound variable names.
As usual there are notions of capture-avoiding substitution
and we write $s^x_t$ to be the result of subsituting $t$ for $x$ in $s$.
We have the usual notions of $\beta$-conversion and $\eta$-conversion:
$(\lambda x.s)t$ $\beta$-reduces to $s^x_t$
and $(\lambda x.sx)$ $\eta$-reduces to $s$ if $x$ is not free in $s$.
The relation $s\sim_{\beta\eta} t$ on terms $s,t\in\Lambda_{\sigma}$
is the least congruence relation closed under $\beta$-conversion and $\eta$-conversion.

The underlying deduction system for Egal is natural deduction with proof terms.
We do not discuss proof terms here, but give the corresponding natural deduction calculus without
proof terms in Figure~\ref{egal:nd}.
The calculus defines when $\Gamma\vdash \varphi$ is derivable
where $\Gamma$ is a finite set of propositions and $\varphi$ is a proposition.

\begin{figure}
\vspace{-0.3cm}
  \begin{center}
    \begin{mathpar}
      \inferrule*[left={Ax}]{\varphi\in\cA}{\Gamma\vdash \varphi}
      \and
      \inferrule*[left={Hyp}]{\varphi\in\Gamma}{\Gamma\vdash \varphi}
      \and
      \inferrule*[left=$\beta$]{\Gamma\vdash \psi \\ \psi \sim_{\beta\eta} \varphi}{\Gamma \vdash \varphi}
      \and
      \inferrule*[left={$\limplies$I}]{\Gamma\cup\{\varphi\} \vdash \psi}{\Gamma\vdash \varphi\limplies \psi}
      \and
      \inferrule*[left={$\limplies$E}]{\Gamma\vdash \varphi\limplies \psi \\ \Gamma\vdash \varphi}{\Gamma \vdash \psi}
      \and
      \inferrule*[left={$\forall$I}]{\Gamma \vdash\varphi^x_y \\ y\in\cV_\sigma {\mbox{ is not free in }} \Gamma\cup\{\varphi\}}{\Gamma\vdash \forall x:\sigma.\varphi}
      \and
      \inferrule*[left={$\forall$E}]{\Gamma \vdash \forall x:\sigma.\varphi\\ t \in\Lambda_\sigma}{\Gamma\vdash \varphi^x_t}
    \end{mathpar}
  \end{center}
  \vspace{-0.2cm}
  \caption{Natural deduction sytem}\label{egal:nd}
\end{figure}

We now briefly discuss the role of polymorphism in Egal.
We have already seen examples where type variables would
be useful. Instead of having infinitely many constants $\varepsilon_\sigma$
in the implementation there is one constant $\varepsilon$
which must be associated with a type when used.
Likewise, the axioms of functional extensionality and choice
make use of type variables and whenever these axioms are used
the instantiations for these type variables must be given.
Some definitions (such as equality, existential quantification and if-then-else)
as well as some theorems (such as the existential introduction rule)
also make use of type variables.
From the beginning Egal was designed to discourage the use of type variables
in the hope of eventually eliminating them.
For this reason constants, definitions, axioms and theorems can use at most
three type variables.
To make this precise we have
three fixed type variables $\nu_0$, $\nu_1$ and $\nu_2$.
For $n\in\{0,1,2,3\}$ we have $\cT^n$ as the set of types freely generated from
$\nu_0|\cdots|\nu_{n-1}|o|\iota|\sigma\to\tau$.
Similarly we have four families of terms $(\Lambda^n_\sigma)_{\sigma\in \cT^n}$
and four judgments $\Gamma\vdash_n\varphi$
where $\Gamma$ is a finite subset of $\Lambda^n_o$
and $\varphi$ is in $\Lambda^n_o$.
All definitions and theorems (with proofs) are given in some type context
determined by $n\in\{0,1,2,3\}$. The context remains fixed throughout
the declaration. If $n > 0$, then when the definition or theorem is used
later (in type context $m\in\{0,1,2,3\}$) it must be given along with $n$ (explicitly given) types
from $\cT^m$ which are used to instantiate the type variables.


In addition to the constants and axioms of the system,
we import a number of constructions and results from the library distributed with Egal.
Some of the constructions are definitions of logical connectives, equality and existential quantification
as well as basic theorems about their properties. Negation of equality, negation of set membership
and subset are imported, defined in the obvious ways. We use the notation $s\not=t$, $s\not\in t$ and $s\subseteq t$
for the corresponding propositions.
The definitions ${\mathsf{TransSet}}$,
${\mathsf{Union\_closed}}$,
${\mathsf{Power\_closed}}$,
${\mathsf{Repl\_closed}}$
and ${\mathsf{ZF\_closed}}$ are imported.
In addition the following definitions are imported:\vspace{-0.05cm}%
\begin{itemize}
\item ${\mathsf{ordinal}}:\iota\to o$ is $\lambda \alpha.{\mathsf{TransSet}}~\alpha\land \forall \beta\in\alpha.{\mathsf{TransSet}}~\beta$. Informally we say {\emph{$\beta$ is an ordinal}} to mean ${\mathsf{ordinal}}~\beta$.
\item ${\mathsf{famunion}}:\iota\to (\iota\to\iota)\to \iota$ is $\lambda X F.\bigcup \{Fx|x\in X\}$. We write $\bigcup_{x\in s} t$ for ${\mathsf{famunion}}~s~(\lambda x.t)$.
\end{itemize}\vspace{-0.05cm}%
We also import the following objects in an opaque way, so that we will only be able to use properties imported
from the library and not the actual definitions.\vspace{-0.05cm}
\begin{itemize}
\item ${\mathsf{Sep}}:\iota\to(\iota\to o)\to \iota$. We write $\{x\in X|\varphi\}$ for ${\mathsf{Sep}}~X~(\lambda x.\varphi)$.
  Results are imported to ensure $\forall z.z\in\{x\in X|\varphi\}\lequiv z\in X\land \varphi^x_z$ is provable.
\item ${\mathsf{ReplSep}}:\iota\to(\iota\to o)\to (\iota\to\iota)\to \iota$. We write $\{s|x\in X {\mathrm{~such~that~}}\varphi\}$ for ${\mathsf{ReplSep}}~X~(\lambda x.\varphi)~(\lambda x.s)$.
  Results are imported to ensure the provability of
  $\forall z.z\in\{s|x\in X{\mathrm{~such~that~}}\varphi\}\lequiv \exists y\in X. \varphi^x_y \land z = s^x_y$.
\item ${\mathsf{UPair}}:\iota\to\iota\to\iota$. We write $\{x,y\}$ for ${\mathsf{UPair}}~x~y$.
  Results are imported to ensure $\forall z.z\in\{x,y\}\lequiv z = x\lor z = y$ is provable.
\item ${\mathsf{Sing}}:\iota\to\iota$. We write $\{x\}$ for ${\mathsf{Sing}}~x$.
  Results are imported to ensure $\forall z.z\in\{x\}\lequiv z = x$ is provable.
\item ${\Inrec}:(\iota\to(\iota\to\iota)\to\iota)\to\iota\to\iota$.
  The ${\Inrec}$ operator is used to define functions by $\in$-recursion over the universe.
  Given a function $F:\iota\to (\iota\to\iota)\to \iota$ satisfying certain conditions,
  $\Inrec~F$ yields a function $f$ satisfying $f~X~=~F~X~f$.
  Its construction is discussed in detail in~\cite{Brown2015}. It is obtained by
  defining the graph of $\Inrec~F$ as the least relation satisfying appropriate closure
  properties and then using $\in$-induction to prove (under appropriate assumptions)
  that this yields a functional relation.
  Here we will only
  need the fundamental property imported as Proposition~\ref{prop:inreceq} below.
  Its use will be essential in proving Tarski's Axiom A in Section~\ref{sec:axioma}.
\end{itemize}\vspace{-0.05cm}
We will freely make use of these imported terms to form new terms below.

Less than 60 results proven in the library need to be imported
in order to prove the results discussed in this paper.
Most of those results are basic results about logic and set theory
and we will leave them implicit here.
The choice axiom and the extensionality axioms
make the logic extensional and classical~\cite{Diaconescu75}.
We import excluded middle and the double negation law
from the library.

The following imported results are worth making explicit:
\begin{proposition}\label{prop:Inirref}  $\forall x.x\notin x$.
\end{proposition}

\begin{proposition}[Regularity]\label{prop:Regularity} $\forall X x.x\in X \limplies \exists Y\in X.\neg\exists z\in X.z \in Y$.
\end{proposition}

\begin{proposition}\label{prop:ordinalHered} $\forall \alpha.{\mathsf{ordinal}}~\alpha\limplies \forall \beta\in\alpha.{\mathsf{ordinal}}~\beta$.
\end{proposition}

\begin{proposition}\label{prop:ordinalTrichotomyOr} $\forall \alpha \beta.{\mathsf{ordinal}}~\alpha \limplies {\mathsf{ordinal}}~\beta\limplies \alpha\in\beta \lor \alpha = \beta \lor \beta\in\alpha$.
\end{proposition}

The fundamental property of $\Inrec$ is imported from the library:
\begin{proposition}[cf. Theorem 1 in~\cite{Brown2015}]\label{prop:inreceq}\vspace{-0.05cm}%
  $$
  \begin{array}{c}
    \forall \Phi:\iota\to(\iota\to\iota)\to\iota.
    (\forall X.\forall g h:\iota\to\iota.(\forall x \in X. g x = h x)\limplies \Phi~X~g = \Phi~X~h) \\
    \to \forall X.\Inrec~\Phi~X = \Phi~X~(\Inrec~\Phi)
  \end{array}
  $$\vspace{-0.05cm}%
\end{proposition}

\section{Higher-order vs. First-order Representations}\label{sec:hofo}

For many concepts we cannot directly compare the formulations in Egal
with those from Mizar since Egal is higher-order. On the other hand,
for the cases of interest in this paper we show we can find
first-order formulations which are provably equivalent in Egal
and have counterparts in Mizar. In particular we will use
this to compare Grothendieck universes in Egal (defined using closure
under replacement) and Grothendieck universes in Mizar (defined using
closure under unions of families of sets).

Tarski's Axiom A (Figure~\ref{fig:mizartarskia}) informally states that
every set is in a Tarski universe.
The most interesting condition in the definition of a Tarski universe
is that every subset of the universe is either a member of the
universe or is equipotent with the universe.
The notion of equipotence of two sets can be represented in different
ways. In first-order one can define when sets $X$ and $Y$ are equipotent as follows:
there is a set $R$ of Kuratowski pairs which essentially encodes
the graph of a bijection from $X$ to $Y$.
In order to state Axiom A in Mizar, one must first define Kuratowski pairs
and then equipotence.
This first-order definition of equipotence can of course be made in Egal
as well.
We omit the details, except to say we easily obtain an
Egal abbreviation ${\mathsf{equip}}$ of type $\iota\to\iota\to o$
with a definition analogous to the definition of equipotence in Mizar.

There is an alternative way to characterize equipotence in Egal without
relying on the set theoretic encoding of pairs and functions.
We simply use functions of type $\iota\to\iota$ given by the underlying simple type theory.

Let ${\mathsf{bij}}:\iota\to\iota\to(\iota\to\iota)\to o$ be\vspace{-0.05cm}
$$
\begin{array}{c}
  \lambda X Y.\lambda f:\iota\to\iota. (\forall u\in X.fu\in Y) \land (\forall u v\in X.fu=fv\limplies u=v) \\
  \land (\forall w \in Y.\exists u\in X.f u =w).
\end{array}
$$\vspace{-0.05cm}
Informally we say {\emph{$f$ is a bijection taking $X$ onto $Y$}} to mean ${\mathsf{bij}}~X~Y~f$.

It is straightforward to prove ${\mathsf{equip}}~X~Y \lequiv \exists f:\iota\to\iota.{\mathsf{bij}}~X~Y~f$ in Egal.
When proving Axiom A in Egal (see Theorem~\ref{thm:tarskia})
we will use $\exists f:\iota\to\iota.{\mathsf{bij}}~X~Y~f$
to represent equipotence. To obtain the first-order formulation Axiom A,
the equivalence of the two formulations of equipotence can be used.


A similar issue arises when considering the notion of being ZF-closed in Mizar.
The definition of ${\mathsf{ZF\_closed}}$
relies on ${\mathsf{Repl\_closed}}$.
${\mathsf{Repl\_closed}}$ relies on the higher-order ${\mathsf{Repl}}$
operator and quantifies over the type $\iota\to\iota$.
An alternative first-order definition of $U$ being ZF-closed is to
say $U$ is $\wp$-closed and $U$ is closed under internal family unions.
The internal family union of a set $I$ and a set $f$
is defined as the set ${\mathsf{famunionintern}}~I~f$
such that $w\in {\mathsf{famunionintern}}~I~f$ if
and only if $\exists i\in I.\exists X.[i,X]\in f\land w\in X$
where $[i,X]$ is the Kuratowski pair $\{\{i\},\{i,X\}\}$.
It is easy to prove such a set exists, in both Egal and Mizar.
Closure of $U$ under internal family unions states that if
$I\in U$, $f$ is a set of Kuratowski pairs representing the
graph of a function from $I$ into $U$,
then ${\mathsf{famunionintern}}~I~f\in U$.

We say $U$ is {\emph{ZF-closed in the FO sense}}
if $U$ is $\wp$-closed and closed under internal family unions.
In Egal it is straightforward to prove
that for transitive sets $U$, $U$ is ZF-closed
if and only if $U$ is ZF-closed in the FO sense.
Grothendieck universes in Egal are transitive
ZF-closed sets.
Grothendieck universes in Mizar are transitive sets
that are ZF-closed in the FO sense.
By the equivalence result, we know these two notions
of Grothendieck universes are equivalent in Egal.

\section{Tarski's Axiom A in Egal}\label{sec:axioma}

We will now describe the HOTG proof of Tarski's Axiom A in Egal.

We begin by using the recursion operator to define an operator returning the set of all sets up to
a given rank:
${\mathsf{V}}:\iota\to\iota$ is $\Inrec (\lambda X v.\bigcup_{x\in X}\wp(v x))$.
  We will write ${\mathbf{V}}_X$ for ${\mathsf{V}}$ applied to $X$.
Using Proposition~\ref{prop:inreceq} it is easy to prove the following:
\begin{theorem}\label{thm:Veq}
  $\forall X.{\mathbf{V}}_X = \bigcup_{x\in X}.\wp({\mathbf{V}}_x)$
\end{theorem}
It is then straightforward to prove a sequence of results.
\begin{theorem}\label{thm:Vfacts} The following facts hold.
 \vspace{-0.05cm}
  \begin{enumerate}
  \item\label{thm:VI} $\forall y x X.x\in X\limplies y\subseteq {\mathbf{V}}_x \limplies y \in {\mathbf{V}}_X$.
  \item\label{thm:VE} $\forall y X.y\in {\mathbf{V}}_X \limplies \exists x\in X. y \subseteq {\mathbf{V}}_x$.
  \item\label{thm:VSubq} $\forall X.X\subseteq {\mathbf{V}}_X$.
  \item\label{thm:VSubq2} $\forall X Y.X\subseteq {\mathbf{V}}_Y \limplies {\mathbf{V}}_X \subseteq {\mathbf{V}}_Y$.
  \item\label{thm:VIn} $\forall X Y.X\in {\mathbf{V}}_Y \limplies {\mathbf{V}}_X \in {\mathbf{V}}_Y$.
  \item\label{thm:VInOrSubq} $\forall X Y.X\in {\mathbf{V}}_Y \lor {\mathbf{V}}_Y \subseteq {\mathbf{V}}_X$.
  \item\label{thm:VInOrSubq2} $\forall X Y.{\mathbf{V}}_X\in {\mathbf{V}}_Y \lor {\mathbf{V}}_Y \subseteq {\mathbf{V}}_X$.
  \end{enumerate}
\end{theorem}
\begin{proof}
  Parts~\ref{thm:VI} and~\ref{thm:VE} are easy consequences of Theorem~\ref{thm:Veq} and properties
  of powersets and family unions. Part~\ref{thm:VSubq} follows by $\in$-induction using Part~\ref{thm:VI}.
  Part~\ref{thm:VSubq2} also follows by $\in$-induction using Parts~\ref{thm:VI} and~\ref{thm:VE}.
  Part~\ref{thm:VIn} follows easily from Parts~\ref{thm:VI}, \ref{thm:VE} and~\ref{thm:VSubq2}.
  Part~\ref{thm:VInOrSubq} follows by $\in$-induction using classical reasoning and Parts~\ref{thm:VI} and \ref{thm:VE}.
  Part~\ref{thm:VInOrSubq2} follows from Part~\ref{thm:VIn} and~\ref{thm:VInOrSubq}.
\end{proof}

Let ${\mathsf{V\_closed}}$ of type $\iota\to o$ be
$\lambda U.\forall X\in U.{\mathbf{V}}_X\in U$.
Informally we say $U$ if ${\mathbf{V}}$-closed to mean ${\mathsf{V\_closed}}~U$.
The following theorem is easy to prove by $\in$-induction using Theorem~\ref{thm:Veq}.
\begin{theorem}\label{thm:VUIn}
  If $U$ is transitive and ZF-closed,
  then $U$ is ${\mathbf{V}}$-closed.
\end{theorem}

Using the choice operator it is straightforward to
construct the inverse of a bijection taking $X$ onto $Y$
and obtain a bijection taking $Y$ onto $X$.
\begin{theorem}\label{thm:bijinv}
  $\forall X Y.\forall f:\iota\to\iota.{\mathsf{bij}}~X~Y~f\limplies {\mathsf{bij}}~Y~X~(\lambda y.\varepsilon x\in X.fx=y)$.
\end{theorem}

We now turn to the most complicated Egal proof.
More than half of the file ending with the proof of
Axiom A is made up of the proof of Lemma~\ref{mainlem:tarskia}.
We outline the proof here
and make some comments about the corresponding formal proof in Egal along the way.
For the full proof see the technical report~\cite{BP2019techreport} or the Egal formalization.


\begin{lemma}\label{mainlem:tarskia}
  Let $U$ be a ZF-closed transitive set
  and $X$ be such that $X\subseteq U$ and $X\not\in U$.
  There is a bijection $f:\iota\to\iota$
  taking $\{\alpha\in U|{\mathsf{ordinal}}~\alpha\}$ onto $X$.
\end{lemma}
\begin{proof}
  In the Egal proof we begin by introducing the local names $U$ and $X$ and making
  the corresponding assumptions.
\begin{verbatim}
let U. assume HT: TransSet U. assume HZ: ZF_closed U.
let X. assume HXsU: X c= U. assume HXniU: X /:e U.
\end{verbatim}
  We next make six local abbreviations. Let
  \begin{itemize}
  \item ${\boldsymbol{\lambda}}$ be $\{\alpha\in U|{\mathsf{ordinal}}~\alpha\}$,
  \item ${\mathbf{P}}:\iota\to\iota\to(\iota\to\iota)\to o$
    be $\lambda \alpha x f.x\in X \land \forall \beta\in\alpha.f\beta\not= x$,
  \item ${\mathbf{Q}}:\iota\to(\iota\to\iota)\to\iota\to o$
    be
    $\lambda \alpha f x.{\mathbf{P}}~\alpha~x~f\land \forall y.{\mathbf{P}}~\alpha~y~f\limplies {\mathbf{V}}_x\subseteq {\mathbf{V}}_y$,
  \item ${\mathbf{F}}:\iota\to(\iota\to\iota)\to \iota$
    be
    $\lambda \alpha f.\varepsilon x.{\mathbf{Q}}~\alpha f x$,
  \item ${\mathbf{f}}:\iota\to\iota$ be $\Inrec{{\mathbf{F}}}$ and
  \item ${\mathbf{g}}:\iota\to\iota$ be $\lambda y.\varepsilon \alpha\in{\boldsymbol{\lambda}}.{\mathbf{f}}\alpha = y$.
  \end{itemize}
  In the Egal proof three of these local definitions are given as follows:
 \vspace{-0.05cm}
 \begin{verbatim}
set lambda : set := {alpha :e U|ordinal alpha}.
...
set f : set->set := In_rec F.
set g : set->set := fun y => some alpha :e lambda, f alpha = y.
\end{verbatim}
 \vspace{-0.05cm}
The following claims are then proven:
\begin{center}
\vspace{-0.25cm}
  \begin{tabular}{cc}
    \parbox{5.5cm}{
  \begin{eqnarray}
    \forall \alpha.{\mathbf{f}}\alpha = {\mathbf{F}}~\alpha~{\mathbf{f}}
    \label{eqn:Lfeq} \\
    \forall \alpha\in{\boldsymbol{\lambda}} . {\mathbf{Q}}~\alpha~{\mathbf{f}}~({\mathbf{f}}\alpha)
    \label{eqn:Lone} \\
    \forall \alpha\in{\boldsymbol{\lambda}} . {\mathbf{f}}\alpha\in X
    \label{eqn:Ltwo} \\
    \forall \alpha \beta\in{\boldsymbol{\lambda}} . {\mathbf{f}}\alpha =  {\mathbf{f}}\beta \limplies \alpha = \beta
    \label{eqn:Lthree}
  \end{eqnarray}
  }
  &
    \parbox{5.5cm}{
  \begin{eqnarray}
    {\mathsf{bij}}~\{{\mathbf{f}}~\alpha|\alpha\in{\boldsymbol{\lambda}}\}~{\boldsymbol{\lambda}}~{\mathbf{g}}
    \label{eqn:Lgbij} \\
    {\boldsymbol{\lambda}} = \{{\mathbf{g}}~y|y\in \{{\mathbf{f}}~\alpha|\alpha\in{\boldsymbol{\lambda}}\}\}
    \label{eqn:Lglam} \\
    \forall x\in X.\exists \alpha\in{\boldsymbol{\lambda}}.{\mathbf{f}} \alpha = x
    \label{eqn:Lfour}
  \end{eqnarray}
    }
  \end{tabular}
  \vspace{-0.2cm}
  \end{center}
  Note that (\ref{eqn:Ltwo}),
  (\ref{eqn:Lthree}) and (\ref{eqn:Lfour})
  imply ${\mathbf{f}}$ is a bijection taking ${\boldsymbol{\lambda}}$ onto $X$,
  which will complete the proof.
  Here we only describe the proof of (\ref{eqn:Lone})
  in some detail and make brief remarks
  about the proofs of the other cases.
  For example, Proposition~\ref{prop:inreceq} is used to prove (\ref{eqn:Lfeq}).

  In the Egal proof we express (\ref{eqn:Lone}) as a claim followed by its subproof.
 \vspace{-0.05cm}
 \begin{verbatim}
claim L1: forall alpha :e lambda, Q alpha f (f alpha).
\end{verbatim}
 \vspace{-0.05cm}
  The subproof is by $\in$-induction.
  Let $\alpha$ be given and assume as inductive hypothesis
  $\forall \gamma.\gamma\in\alpha\limplies \gamma\in{\boldsymbol{\lambda}} \limplies {\mathbf{Q}}~\gamma~{\mathbf{f}}~({\mathbf{f}}\gamma)$.
  Assume $\alpha\in{\boldsymbol{\lambda}}$, i.e., $\alpha \in U$ and ${\mathsf{ordinal}}~\alpha$.
  Under these assumptions
  we can prove the following subclaims:
  \begin{center}
  \vspace{-0.25cm}
    \begin{tabular}{ccc}
      \parbox{4.2cm}{
  \begin{eqnarray}
    \forall \beta \in \alpha. {\mathbf{Q}}~\beta~{\mathbf{f}}~({\mathbf{f}} \beta)
    \label{eqn:Lonea} \\
    \forall \beta \in \alpha. {\mathbf{f}} \beta \in X
    \label{eqn:Loneaone} 
  \end{eqnarray}
  }
  &
      \parbox{3.8cm}{
  \begin{eqnarray}
    \{{\mathbf{f}}\beta | \beta\in\alpha\} \subseteq X
    \label{eqn:Lonebaa} \\
    \{{\mathbf{f}}\beta | \beta\in\alpha\} \in U
    \label{eqn:Lonebb} \\
    \exists x. {\mathbf{P}}~\alpha~x~{\mathbf{f}}
    \label{eqn:Loneb} 
  \end{eqnarray}
  }
  &
      \parbox{3.5cm}{
  \begin{eqnarray}
    \exists x. {\mathbf{Q}}~\alpha~{\mathbf{f}}~x
    \label{eqn:Lonec} \\
    {\mathbf{Q}}~\alpha~{\mathbf{f}}~({\mathbf{F}}~\alpha~{\mathbf{f}})
    \label{eqn:Loned}
  \end{eqnarray}
  }
    \end{tabular}
    \vspace{-0.2cm}
  \end{center}
  We show only the proof of ({\ref{eqn:Lonec}}) assuming
  we have already established ({\ref{eqn:Loneb}}).
  Let ${\mathbf{Y}}$ be
  $\{{\mathbf{V}}_x |x\in X {\mathrm{~such~that~}} \forall \beta\in \alpha.{\mathbf{f}}\beta\not=x\}$.
  By ({\ref{eqn:Loneb}}) there is a $w$
  such that ${\mathbf{P}}~\alpha~w~f$.
  That is, $w\in X$ and $\forall \beta\in\alpha.{\mathbf{f}}\beta\not= w$.
  Clearly ${\mathbf{V}}_w\in {\mathbf{Y}}$.
  By Regularity (Proposition~\ref{prop:Regularity})
  there is some $Z\in {\mathbf{Y}}$ such that
  $\neg\exists z \in {\mathbf{Y}}.z\in Z$.
  Since $Z\in {\mathbf{Y}}$ there must be some $x\in X$
  such that $Z = {\mathbf{V}}_x$ and $\forall \beta\in\alpha.{\mathbf{f}}\beta\not=x$.
  We will prove ${\mathbf{Q}}~\alpha~{\mathbf{f}}~x$ for this $x$.
  We know ${\mathbf{P}}~\alpha~x~{\mathbf{f}}$ since
  $x\in X$ and $\forall \beta\in\alpha.{\mathbf{f}}\beta\not=x$.
  It remains only to prove
  $\forall y.{\mathbf{P}}~\alpha~y~{\mathbf{f}}\limplies {\mathbf{V}}_x\subseteq {\mathbf{V}}_y$.
  Let $y$ such that ${\mathbf{P}}~\alpha~y~{\mathbf{f}}$ be given.
  By Theorem~\ref{thm:Vfacts}:\ref{thm:VInOrSubq2} either ${\mathbf{V}}_y\in {\mathbf{V}}_x$
  or ${\mathbf{V}}_x\subseteq {\mathbf{V}}_y$.
  It suffices to prove ${\mathbf{V}}_y\in {\mathbf{V}}_x$ yields a contradiction.
  We know ${\mathbf{V}}_y \in {\mathbf{Y}}$ since ${\mathbf{P}}~\alpha~y~{\mathbf{f}}$.
  If ${\mathbf{V}}_y\in {\mathbf{V}}_x$, then
  ${\mathbf{V}}_y \in Z$ (since $Z = {\mathbf{V}}_x$),
  contradicting
  $\neg\exists z \in {\mathbf{Y}}.z\in Z$.

  We conclude ({\ref{eqn:Loned}}) by ({\ref{eqn:Lonec}}) and the property
  of the choice operator used in the definition of ${\mathbf{F}}$.
  By ({\ref{eqn:Loned}}) and ({\ref{eqn:Lfeq}}) we have
  ${\mathbf{Q}}~\alpha~{\mathbf{f}}~({\mathbf{f}}\alpha)$.
  Recall that this was proven under an inductive hypothesis for $\alpha$.
  We now discharge this inductive hypothesis and conclude ({\ref{eqn:Lone}}).

  One can easily prove (\ref{eqn:Ltwo}) and (\ref{eqn:Lthree}) from (\ref{eqn:Lone})
  and Proposition~\ref{prop:ordinalTrichotomyOr}.
  From (\ref{eqn:Lthree}) and Theorem~\ref{thm:bijinv} we have (\ref{eqn:Lgbij})
  and from this we obtain (\ref{eqn:Lglam}).

  Finally to prove (\ref{eqn:Lfour}) assume there is
  some $x\in X$ such that
  $\neg\exists\alpha \in {\boldsymbol{\lambda}}. {\mathbf{f}}\alpha = x$.
  Under this assumption one can prove ${\boldsymbol{\lambda}} \in {\boldsymbol{\lambda}}$, contradicting Proposition~\ref{prop:Inirref}.
  It is easy to prove ${\boldsymbol{\lambda}}$ is an ordinal,
  so it suffices to prove ${\boldsymbol{\lambda}} \in U$.
  The proof that ${\boldsymbol{\lambda}} \in U$
  makes use of
  Proposition~\ref{prop:ordinalHered}, Theorem~\ref{thm:VUIn},
  Theorem~\ref{thm:Vfacts}:\ref{thm:VSubq},
  (\ref{eqn:Lone}), (\ref{eqn:Lgbij}), (\ref{eqn:Lglam})
  and the closure properties of $U$.
\end{proof}

We can now easily conclude Tarski's Axiom A in Egal.

\begin{theorem}[Tarski A]\label{thm:tarskia}
  For each set $N$ there exists an $M$ such that%
  \begin{enumerate}
  \item $N\in M$,
  \item $\forall X\in M.\forall Y\subseteq X.Y\in M$,
  \item $\forall X\in M.\exists Z\in M.\forall Y \subseteq X. Y\in Z$ and
  \item $\forall X\subseteq M.(\exists f:\iota\to\iota.{\mathsf{bij}}~X~M~f)\lor X\in M$.
  \end{enumerate}
\end{theorem}%
\begin{proof}
  We use $U := \UnivOf{N}$ as the witness for $M$.
  We know $N\in\UnivOf{N}$, $\UnivOf{N}$ is transitive and ZF-closed by
  the axioms of our set theory.
  All the properties except the last follow easily from these facts.
  We focus on the last property.
  Let $X\subseteq U$ be given. Since we are in a classical
  setting it is enough to assume $X\notin U$ and
  prove there is some bijection $f:\iota\to\iota$ taking $X$ onto $U$.
  Since $U\subseteq U$ and $U\notin U$ (using Proposition~\ref{prop:Inirref}),
  we know there is a bijection $g$ taking $\{\alpha\in U|{\mathsf{ordinal}}~\alpha\}$ onto $U$ by
  Lemma~\ref{mainlem:tarskia}.
  Since $X\subseteq U$ and $X\notin U$,
  we know there is a bijection $h$ taking $\{\alpha\in U|{\mathsf{ordinal}}~\alpha\}$ onto $X$ by
  Lemma~\ref{mainlem:tarskia}.
  By Theorem~\ref{thm:bijinv} there is a bijection $g^{-1}$ taking
  $X$ onto $\{\alpha\in U|{\mathsf{ordinal}}~\alpha\}$.
  The composition of $g^{-1}$ and $h$ yields a bijection $f$ taking
  $X$ onto $U$ as desired.
\end{proof}

\section{Grothendieck Universes in Mizar}\label{sec:grothuniv}

In this section we construct Grothendieck universes
using notions
introduced in the MML articles
\texttt{CLASSES1} and \texttt{CLASSES2}~\cite{CLASSES1.ABS,CLASSES2.ABS}.
For this purpose, first, we
briefly introduce the relevant constructions from these articles. 
We then define the notion of a Grothendieck universe of a set $A$
as a Mizar type, the type of all transitive sets
with $A$ as a member
that are closed under power sets and internal family unions.
Since Mizar types must be nonempty, we are required to construct
such a universe.
We finally introduce a functor \miz{GrothendieckUniverse'}$A$ that
returns the least set of the type. 
Additionally, we show that every such Grothendieck universe is closed under replacement
formulating the property as a Mizar scheme.%

To simplify notation we present selected Mizar operators
in more natural ways closer to informal mathematical practice.
In particular, we use $\emptyset$, $\in$, $\subseteq$, $\wp$, $|\cdot|$, $\bigcup$ to represent
Mizar symbols as \miz{{}}, \miz{in}, \miz{c=}, \miz{bool},
\miz{card}, \miz{union}, respectively.

Following Bancerek, we will start with the construction of the least Tarski universe
that contains a given set $A$.
Tarski's Axiom A directly implies that there exists a \miz{Tarski'set} $T_A$ that contains $A$
where \miz{Tarski} 
is a Mizar {\it attribute} (for more details see~\cite{JFR1980}) defined as follows:
\vspace{-0.03cm}%
\begin{lstlisting}
  attr'T'is'Tarski'means'::'CLASSES1:def 2
    T'is'subset-closed'&'(for'X'holds'X'$\in$'T'implies'$\wp$(X)$\in$'T)'&
    for'X'holds'X'$\subseteq$'T'implies'X,T'are_equipotent'or'X'$\in$'T;
\end{lstlisting}\vspace{-0.04cm}
Informally we say that $T$ is \miz{Tarski} to mean $T$ is closed under subset, power sets and each subset of $T$
 is a member of $T$ or is equipotent with $T$.
Then one shows that $\bigcap\{X|A\in X \subseteq T_A, X\mbox{ is } \miz{Tarski}\,\miz{set}\}$
is the least (with respect to inclusion) \miz{Tarski'set} that contains $A$, denoted by
\miz{Tarski-Class'}$A$.

By definition it is easy to prove the following:\vspace{-0.02cm}
\begin{theorem}\label{thm:TCfacts} The following facts hold.\vspace{-0.05cm}%
  \begin{enumerate}
  \item\label{thm:TC1} $\forall A.\: A \in \miz{Tarski-Class}\,A$,
  \item\label{thm:TC2} $\forall A\,X\,Y.\: Y \subseteq X \land X \in \miz{Tarski-Class}\, A\limplies Y \in \miz{Tarski-Class}\, A$,
  \item\label{thm:TC3} $\forall A\,X.\: Y \in \miz{Tarski-Class'} A\limplies \wp(X)\in \miz{Tarski-Class'} A$,
  \item\label{thm:TC4} $\forall A\,X.\: X \subseteq \miz{Tarski-Class}\,A\land |X| < |\miz{Tarski-Class}\,A|
  \limplies X \in \miz{Tarski-Class}\,A$.\vspace{-0.02cm}%
   \end{enumerate}
\end{theorem}

\vspace{-0.05cm}
Tarski universes, as opposed to Grothendieck universes,
might not be transitive (called \miz{epsilon-transitive} in the MML)
but via transfinite induction. 
By Theorems 22 and 23 in~\cite{CLASSES1.ABS} we know
$\miz{Tarski-Class'}A$ is transitive
if $A$ is transitive.
Therefore, in our construction we
take the transitive closure of $A$ prior to the application of the
\miz{Tarski-Class} functor.
Using a recursion scheme we know
for a given set $A$
there exists a recursive sequence $f$ such that
$f(0)=A$ and $\forall k\in \mathbb{N}.\:f(k+1)=\bigcup f(k)$.
For such an $f$, $\bigcup \{f(n)|n \in \mathbb{N}\}$
is the least (with respect to the inclusion) transitive set that includes $A$
(or contains $A$ if we start with  $f(0)=\{A\}$).
The operator is defined in~\cite{CLASSES1.ABS} as follows:\vspace{-0.035cm}%
\begin{lstlisting}
  func'the_transitive-closure_of'A'->'set'means'::'CLASSES1:def'7
    for'x'holds'x'$\in$'it'iff'ex'f'being'Function,'n'being'Nat'st
      x'$\in$'f.n'&'dom'f'='$\mathbb{N}$'&'f.0'='A'&'for'k'being'Nat'holds'f.(k+1)'='$\bigcup$'f.k;
\end{lstlisting}\vspace{-0.025cm}%

We now turn to a formulation of ZF-closed property in Mizar.
It is obvious that $\wp$-closed, $\bigcup$-closed properties can we expressed as
two Mizar types as follows:\vspace{-0.04cm}%
\begin{lstlisting}
  attr'X'is'power-closed'means'for'A'being'set'st'A'$\in$'X'holds'$\wp$(A)'$\in$'X;
  attr'X'is'union-closed'means'for'A'being'set'st'A'$\in$'X'holds'$\bigcup$(A)'$\in$'X;
\end{lstlisting}\vspace{-0.04cm}%
Note that we cannot express the closure under replacement as a Mizar type
since each condition that occurs after \miz{means} has to be a first-order statement.
We must therefore use an alternative approach that uses closure under
internal family unions using the notion of a function as well as its domain (\miz{dom}) and range (\miz{rng})
as follows:\vspace{-0.045cm}%
\begin{lstlisting}
  attr'X'is'FamUnion-closed'means
    for'A'being'set'for'f'being'Function'st'dom'f'='A'&'rng'f'$\subseteq$'X'&'A'$\in$'X'
      holds'$\bigcup$'rng'f'$\in$'X;
\end{lstlisting}\vspace{-0.045cm}%

Comparing the properties of Tarski and Grothendieck universes we can prove the following:
\begin{theorem}\label{thm:clusters}
The following facts hold.\vspace{-0.035cm}%
\begin{enumerate}
  \item\label{thm:CL1}$\forall X. X\mbox{ is }\miz{Tarski} \limplies X \mbox{ is } \miz{subset-closed'power-closed},$
  \item\label{thm:CL2}$\forall X. X\mbox{ is }\miz{epsilon-transitive'Tarski}\limplies X \mbox{ is } \miz{union-closed} $,
  \item\label{thm:CL3}$\forall X. X\mbox{ is }\miz{epsilon-transitive'Tarski}\limplies X \mbox{ is } \miz{FamUnion-closed}$.\vspace{-0.035cm}%
\end{enumerate}
\end{theorem}
\begin{proof}
Part~\ref{thm:CL1} is an easy consequences of the \miz{Tarski} definition and properties of
powersets. Part~\ref{thm:CL2} is a direct conclusion of
the MML theorem \miz{CLASSES2:59}.
To prove \ref{thm:CL3} let $X$ be an \miz{epsilon-transitive'Tarski'set}, $A$ be a set and
$f$ be a function such that $\miz{dom'} f={A}$, $\miz{rng'} f \subseteq X$, $A \in X$.
Since $X$ is \miz{subset-closed} as a \miz{Tarski'set} and $A \in X$, we know that $\wp (A) \subseteq X$.
By Cantor's theorem we conclude that $|A| < |\wp (A)|$ and consequently $|A| < |X|$.
Since $|\miz{rng'} f| \leq |\miz{dom'} f|=|A|$, we know that  $\miz{rng'} f$ is not equipotent with $X$.
Then $\miz{rng'} f \in X$ since $X$ is \miz{Tarski} and $\miz{rng'} f\subseteq X$, and finally
 $\bigcup \miz{rng'} f \in X$ by Part~\ref{thm:CL2}.
\end{proof}\vspace{-0.11cm}

We can now easily infer from Theorem~\ref{thm:clusters} that
the term:\vspace{-0.055cm}%
\begin{equation}\label{GrothendieckOfTerm}
\miz{Tarski-Class(the_transitive-closure_of'$\{$A$\}$)}
\end{equation}\vspace{-0.055cm}
is suitable to prove that the following Mizar type is inhabited:\vspace{-0.04cm}%
\begin{lstlisting}
  mode'Grothendieck'of'A'->'set'means
     A'$\in$'it'&'it'is'epsilon-transitive'power-closed'FamUnion-closed;
\end{lstlisting}\vspace{-0.04cm}%
Now it is a simple matter to construct
the Grothendieck universe of a given set $A$
(\miz{GrothendieckUniverse'}$A$)
since $\bigcap\{X|X \subseteq G_A, X\mbox{ is } \miz{Grothendieck'of'}A\}$
is the least (with respect to the
inclusion) \miz{Grothendieck'of'}$A$, where $G_A$ denotes the term \eqref{GrothendieckOfTerm}.

As we noted earlier, we cannot express the closure under replacement property
as a Mizar type or even assumption in a Mizar theorem.
However we can express and prove that every \miz{Grothendieck'of'}$A$
satisfies this property as a scheme as follows:\vspace{-0.05cm}%
\begin{lstlisting}
  scheme'ClosedUnderReplacement
      {A()'->'set,'U()'->'Grothendieck'of'A(),F(set)'->'set}:
    {F(x)'where'x'is'Element'of'A():'x'$\in$'A()}'$\in$'U()
  provided
    for'X'being'set'st'X'$\in$'A()'holds'F(X)'$\in$'U()
\end{lstlisting}\vspace{-0.05cm}
The proof uses a function that maps each $x$ in \miz{A()} to \miz{$\{$F(}$x$\miz{)$\}$}.\footnote{Note that in Mizar schemes, schematic variables such as $A$ must be given
as \miz{A()} to indicate $A$ is a term with no dependencies.}

\section{Future Work}\label{sec:futurework}

The present work sets the stage for two future possibilities:
translating Mizar's MML into Egal
and translating Egal developments into Mizar articles.
Translating the MML into Egal is clearly possible
in principle, but will be challenging in practice.
The ``obvious'' inferences allowed by Mizar would
need to be elaborated for Egal. Furthermore,
the implicit inferences done by Mizar's soft typing
system would need to be made explicit for Egal.
A general translation from Egal developments to Mizar articles is not possible in principle
(since Egal is higher-order) although we have shown it is often possible
in practice (by handcrafting equivalent first-order formulations of concepts).
There is no reason to try to translate the small Egal library to Mizar,
but it might be useful to have a partial translation for Egal developments
that remain within the first-order fragment.
With such a translation a user could formalize a mathematical development in Egal
and automatically obtain a Mizar article.

\section{Conclusion}\label{sec:concl}

\vspace{-0.01cm}%
We have presented the foundational work required in order
to port formalizations from Mizar to Egal or Egal to Mizar.
In Egal this required a nontrivial proof of Tarski's Axiom A,
an axiom in Mizar. In Mizar this required finding
equivalent first-order representations for the relevant
higher-order terms and propositions used in Egal
and then constructing a Grothendieck universe operator in Mizar.
\vspace{-0.01cm}%
\section*{Acknowledgment}
\vspace{-0.01cm}%
This work has been supported by the European Research Council (ERC)
Consolidator grant nr. 649043 \textit{AI4REASON}
and the Polish National Science Center
granted by decision n$\!^\circ$DEC-2015/19/D/ST6/01473.

\bibliographystyle{splncs04}

\end{document}